\documentclass[10pt,conference]{IEEEtran}
\IEEEoverridecommandlockouts
\usepackage{blindtext}
\usepackage{amsmath}

\usepackage{hyperref}
\hypersetup{
  colorlinks   = true, 
  urlcolor     = blue, 
  linkcolor    = blue, 
  citecolor   = blue 
}
\usepackage{bm}
\usepackage{gensymb}
\usepackage[normalem]{ulem}
\usepackage{euscript}
\usepackage[pdftex]{graphicx}
\usepackage{xspace}
\usepackage{xfrac}
\usepackage{enumerate}
\usepackage{enumitem}
\usepackage{braket}
\usepackage{float,fancyvrb}
\usepackage[T1]{fontenc}
\usepackage{lmodern}
\usepackage{lipsum} 
\graphicspath{{./Figures/}}
\usepackage{wrapfig}
\usepackage{scrextend}
\usepackage{comment}
\usepackage{url}
\usepackage{amssymb}
\usepackage{xcolor}
\usepackage[noadjust]{cite}

\definecolor{ao}{rgb}{0.0, 0.5, 0.0}
\usepackage{amsthm}
\newtheorem{theorem}{Theorem}
\newtheorem{lemma}{Lemma}
\newtheorem{corollary}{Corollary}[theorem]
\setlist{nosep}
\usepackage{etoolbox}
\makeatletter
\patchcmd{\@makecaption}
  {\scshape}
  {}
  {}
  {}
\makeatletter
\patchcmd{\@makecaption}
  {\\}
  {.\ }
  {}
  {}
\makeatother

\begin{document}

\title{Discrete quadratic model QUBO solution landscapes}

\author{

\IEEEauthorblockN{Tristan Zaborniak}
\IEEEauthorblockA{\textit{Department of Computer Science} \\
\textit{University of Victoria}\\
 Victoria, Canada\\
tristanz@uvic.ca}

\and

\IEEEauthorblockN{Ulrike Stege}
\IEEEauthorblockA{\textit{Department of Computer Science} \\
\textit{University of Victoria}\\
Victoria, Canada\\
ustege@uvic.ca}
}

\maketitle

\begin{abstract}

Many computational problems involve optimization over quadratically-interacting discrete variables. Known as discrete quadratic models (DQMs), these problems in general are NP-hard. Encoding DQMs as quadratic unconstrained binary optimization (QUBO) models permits them to be optimized by quantum and quantum-inspired hardware and  algorithms, which provide alternative means of sampling good solutions versus conventional, classical methods. However, converting DQMs to QUBO models changes the structure of the solution space, and can introduce invalid solutions that must be penalized to ensure that the global optimum is valid. This is accomplished by adding weighted constraints to the QUBO objective function. Here, we investigate the influence of encoding-specific connectivity, dimension, and penalty strength on the structure of QUBO-encoded DQM solution landscapes and their optimization. Specifically, we focus on one-hot and domain-wall encodings, establishing penalty strength thresholds delineating when valid and invalid solutions can be expected to occupy local minima, and outlining differences in their structural characteristics important to the preservation of local minima native to the original problem form.

\end{abstract}

\bigskip


\section{Introduction}\label{introduction}

Optimization of discrete quadratic models (DQMs) covers a large number of combinatorial optimization problems known to be NP-hard \cite{Chaovalitwongse2008, Lucas2014, Lodewijks}, including the quadratic assignment problem, graph problems (including graph coloring and max-cut), a large number of scheduling problems, and the travelling salesperson problem \cite{Lawler1963, ela1998, jensen2011graph, marx2004graph}. Recently, quantum computers and quantum-inspired computers (e.g., quantum annealers, the QAO algorithm executed by gate-model quantum computers, digital annealers) have been applied to solve these problems, reformulated as quadratic unconstrained binary optimization (QUBO) models, motivated by the observation that their specific architectures and solution methods can reduce the time- and space-complexity needed to solve certain problem classes versus conventional, classical methods such as simulated annealing and genetic algorithms \cite{Bian2014, farhi2016quantum, Aramon2019, zahedinejad2017combinatorial, 9045100, Crosson2021}.

Mapping a DQM to a QUBO model requires encoding the discrete variables of the model into binary variables in such a way that only quadratic interactions are permitted between binary variables. Three such encodings are currently known: \textit{one-hot}, \textit{domain-wall}, and \textit{binary} \cite{Berwald2022}. While one-hot and domain-wall encodings of discrete variables are natively quadratic, this is not guaranteed with binary encoding, where auxiliary variables are generally needed to quadratize higher-order terms if the encoding is to be lossless \cite{leib2016transmon, dattani1901quadratization, Tamura2021}. In fact, finding an optimal quadratization which uses the minimum number of auxiliary variables is known to be NP-hard \cite{dridi2018novel}. For this reason, in addition to a recent proof that domain-wall encoding is the most binary variable-efficient QUBO model encoding in the case of a general DQM \cite{Berwald2022}, one-hot and domain-wall encodings are preferred.

A significant artifact of these two encodings concerns their introduction of \textit{invalid} solutions to the QUBO model solution space, which have no meaning or interpretation with respect to the original DQM \cite{invalid} (e.g., a one-hot encoded discrete variable is a register of bits, only one of which can be 1; if there is more or less than one bit equal to 1, this is an invalid solution). As such, it must be ensured that these invalid solutions do not occupy the optimal positions within the QUBO-encoded DQM solution landscape. This is accomplished by the introduction of constraints to the QUBO model objective function, each weighted by a tunable penalty parameter of sufficient strength \cite{Verma2022, Ayodele2022}. Merely satisfying this sufficiency condition might result in the problem being more or less difficult to solve, though, as varying the strength of this penalty parameter drastically changes the structure of the solution landscape \cite{Kauffman1987, Stadler, palmer2018optimization, Garca2022}. 

Additionally, encoding DQMs as QUBO models changes the connectivity between solutions as a result of projection of the original space to a higher-dimensional space, which can alter whether a given element occupying an extremum of the original problem space occupies an extremum of the projected problem space. While this is negligible in the case that only the global optimum is of interest, many use-cases are also concerned with good solutions that are not necessarily globally optimal, in which case encoding choice matters on the basis of its connectivity. Many such projections per encoding exist, corresponding to permuted assignments of discrete variables to valid binary variable strings. To the best of our knowledge, the implications of QUBO-encoded DQM topology on solution landscape structure and optimization have not yet been studied.

In this work we systematically investigate the structure of QUBO-encoded DQM solution landscapes under one-hot and domain-wall encodings in terms of solution landscape connectivity and dimension, and under the influence of varying penalty parameter strength. First we present the one-hot and domain-wall encodings of an arbitrary, unconstrained DQM. Then, we derive and discuss features of their solution landscapes with respect to connectivity, dimension, and penalty parameter strength, focusing on the conditions under which all solutions of a given class (valid or invalid) occupy (or do not occupy) local minima. These conditions are important to establishing penalty strengths that reduce the number of invalid solutions occupying local minima while guaranteeing that valid solutions of interest (which which may comprise a set of solutions of varying, low energies) continue to occupy local minima, a requirement given that optimization algorithms overwhelmingly converge on minima.

Our findings are as follows. For one-hot QUBO-encoded DQMs, no invalid solution occupies a local minimum for sufficiently large penalty strengths. Similarly, all valid solutions occupy local minima for sufficiently large penalty strengths. Furthermore, no valid solution occupies a local minimum for sufficiently small penalty strengths. By contrast, we find for domain-wall QUBO-encoded DQMs that we cannot in general guarantee that all invalid solutions do not occupy local minima, regardless of our selection of penalty strength. Moreover, it is never the case with domain-wall QUBO-encoded DQMs that all valid solutions occupy local minima. Finally, we discuss the significance of understanding solution landscape structures to selection of an appropriate encoding for a given problem. We concentrate this discussion on the encodings investigated here, but emphasize a wider applicability.

\section{DQMs encoded as QUBO models}\label{DQMs as QUBOs}

DQMs are polynomials over discrete variables (numeric or categorical), limited to terms of degree two or less. An arbitrary DQM is expressed as follows:

\begin{equation}
    H_{DQM} = \sum_i A_{(i)}(d_i)+\sum_{i\geq j}B_{(i,j)}(d_i, d_j)
\end{equation}

\noindent where $d_i$ are discrete variables, and $A_{(i)}$ and $B_{(i,j)}$ are real-valued functions over these variables. 

To express a DQM in binary variables, we denote each variable $d_i$ by a set of \textit{sub-variables} $x_{i, \alpha}$, where $i$ refers to the variable index (register) and $\alpha$ to the index of the variable value (state). These $x_{i, \alpha}$ are such that:

\begin{equation}
    x_{i, \alpha} = 
    \begin{cases}
        1, &\text{variable }i\text{ matches value indexed by }\alpha \\
        0, &\text{otherwise}
    \end{cases}
\end{equation}

\noindent The DQM objective function may then be expressed as:

\begin{equation}\label{DQM}
    H_{DQM} = \sum_{i\geq j}\sum_{\alpha,\beta} C_{(i,j,\alpha,\beta)}x_{i,\alpha}x_{j,\beta}
\end{equation}

\noindent where the coefficients $C_{(i,j,\alpha,\beta)}$ express both the linear and quadratic interaction energies between the DQM sub-variables (given that $x_{i,\alpha}x_{i,\alpha}=x_{i,\alpha}$). Throughout this work, for simplicity we let the number of variables be $k$ and the number of values per variable be $l$. 

\begin{table*}
    \centering
    \caption{Various features of one-hot and domain-wall encodings of DQMs. Note that ``adjacent'' means with a Hamming distance of 1.}
    \label{table1}
    \begin{tabular}{ |p{4cm}||p{6cm}|p{6cm}|  }
         \hline
         Feature & One-Hot Encoding & Domain-Wall Encoding\\
         \hline
         binary variables & $kl$ & $k(l-1)$ \\
         pairwise interactions & $kl(kl-1)/2$ & $k(l-1)(k(l-1)-1)/2$ \\
         invalid solutions & $2^{kl}-l^k$ & $2^{k(l-1)}-l^k$ \\
         maximum penalty ($\times\gamma$) & $k(l-1)^2$ & $k\lfloor(l-1)/2\rfloor$ \\
         connectivity & valid solutions not adjacent & valid solutions can be adjacent \\ 
         ordering of states & arbitrary & arbitrary\\
         largest coefficient & $\max_{(i,j,\alpha,\beta)}C_{(i,j,\alpha,\beta)}$ & $\geq\max_{(i,j,\alpha,\beta)}C_{(i,j,\alpha,\beta)}$ \\
        \hline
    \end{tabular}
\end{table*}
\subsection{One-hot encoding}

Per DQM sub-variable $x_{i,\alpha}$, a one-hot QUBO encoding assigns a corresponding binary variable $b_{i, \alpha}$, identical to $x_{i, \alpha}$ except in notation, to emphasize its binary character. The following penalty function, $H_{OH}^P$, then ensures that all invalid solutions, which are with at least one register $i^\prime$ such that $\sum_\alpha b_{i^\prime, \alpha} \neq 1$, are assigned a positive penalty:

\begin{equation}\label{penalty_oh}
    H_{OH}^{P} = \sum_i\Big(\sum_\alpha b_{i, \alpha}-1\Big)^2
\end{equation}


\noindent Note that the penalty per register is proportional to the number of bits with value $1$ (minus one) squared. With $k$ registers of length $l$ each, the magnitude of this penalty therefore ranges from $0$ to $k(l-1)^2$.

This penalty is then added to the objective function while being multiplied by a penalty parameter $\gamma_{OH}$ such that the overall one-hot QUBO-encoded DQM, $H_{OH}$, is:

\begin{align}
    H_{OH} = &\sum_{i\geq j}\sum_{\alpha,\beta} C_{(i,j,\alpha,\beta)}b_{i,\alpha}b_{j,\beta} \nonumber \\ &+ \gamma_{OH}\sum_i\Big(\sum_\alpha b_{i, \alpha}-1\Big)^2
\end{align}

\noindent It is worth noting explicitly that this encoding scheme requires $kl$ binary variables and their associated linear terms, and $kl(kl-1)/2$ interaction terms to represent a DQM of $k$ variables with $l$ possible values each. Moreover, we point out that of the $2^{kl}$ possible binary solution vectors, only $l^k$ are valid in the absence of further constraints, while the remaining $2^{kl}-l^k$ solutions are invalid. Finally, note that valid solutions do not lie adjacent to one another in Hamming space; at least two bit-flips are required to move from one valid solution to another, the first of which takes the original, valid solution to an invalid solution. See Table \ref{table1} for a summary of the one-hot encoding features mentioned in this section.

\subsection{Domain-wall encoding}

Domain-wall encoding is a relatively recent innovation as compared to one-hot encoding for DQMs, and is based on the physics of domain-walls in one-dimensional Ising spin chains \cite{Chancellor2019}. While typically formulated in terms of spin variables $\sigma_{i, \alpha}\in\{-1, +1\}$ before being translated into binary variables, here we directly use binary variables. Domain-wall encoding is a unary encoding, with the value of a register represented by the position of a domain-wall within the register. Binary variables $b_{i,\alpha}$ are defined for $\alpha\in{0,..., l-2}$ (one less index than in the case of one-hot encoding), and the boundary conditions $b_{i, -1}=1$ and $b_{i, l-1}=0$ are enforced per register. Note that these enforced boundary conditions \textit{do not} correspond to physical bits used in computations.

Specifically, we replace each $x_{i,\alpha}$ in Equation \ref{DQM} with $b_{i, \alpha-1}-b_{i, \alpha}$. This transformation is such that valid solutions contain only one non-zero term (one \textit{domain-wall}) per register, while invalid solutions contain more than one domain-wall in at least one register. The penalty function, $H_{DW}^P$, ensuring that all invalid solutions are assigned a positive penalty, is given by Ref.~\cite{Codognet2022} as:

\begin{equation}\label{penalty_dw}
    H_{DW}^{P} = \sum_i\sum_{\alpha}(b_{i, \alpha}-b_{i, \alpha}b_{i, \alpha-1})
\end{equation}

\noindent Note that the penalty per register is proportional to the number of domain walls present (minus one). With $k$ registers each of length $l-1$ (as opposed to length $l$ in the one-hot encoding scheme), the magnitude of this penalty ranges from $0$ to $k\lfloor(l-1)/2\rfloor$.

As with the one-hot penalty function, this domain-wall penalty function is then added to the objective function and multiplied by a penalty parameter $\gamma_{DW}$, giving the overall domain-wall QUBO-encoded DQM, $H_{DW}$, as:

\begin{align}
    H_{DW} = &\sum_{i\geq j}\sum_{\alpha,\beta} C_{(i,j,\alpha,\beta)}(b_{i, \alpha-1}-b_{i, \alpha})(b_{j, \beta-1}-b_{j, \beta}) \nonumber \\&+ \gamma_{DW}\sum_i\sum_{\alpha}(b_{i, \alpha}-b_{i, \alpha}b_{i, \alpha-1})
\end{align}

\noindent With this encoding, we require $k(l-1)$ binary variables and their associated linear terms, and $k(l-1)(k(l-1)-1)/2 = kl(kl-1)/2-k(k+1)/2$ interaction terms, amounting to a savings of $k$ binary variables and $k(k+1)/2$ interaction terms versus a one-hot encoding of the same problem \cite{Codognet2022}. Of the $2^{k(l-1)}$ possible solution vectors, $l^k$ are valid solutions, and the remaining $2^{k(l-1)}-l^k$ are invalid solutions in the absence of further constraints. Moreover, in contrast to one-hot encoding, valid solutions under domain-wall encoding lie adjacent to between $k$ and $2k$ valid solutions (of $kl$ possible 1-local neighbors). Table \ref{table1} contains a summary of the domain-wall encoding features mentioned in this section.

\section{Solution landscape features of DQMs encoded as QUBO models}\label{The effect of penalty strength on solution landscape features}

As described above, we can represent both one-hot and domain-wall encodings of DQMs as a sum of a \textit{cost} function, $c(x)$ (Equation \ref{DQM}), and \textit{penalty} function, $p(x)$ (Equations \ref{penalty_oh} or \ref{penalty_dw}), such that their QUBO model functions, $f(x)$, can then be written as:

\begin{equation}
    f(x) = c(x)+\gamma p(x)
\end{equation}

\noindent where $x \in \{0,1\}^n$ is a binary solution vector of a length appropriate to the encoding under consideration (i.e., in our case $n\in\{kl, k(l-1)\}$), and $\gamma\in\{\gamma_{OH}, \gamma_{DW}\}$. Denoting the optimum valid solution $x^*$, we must select a $\gamma$ to satisfy that $f(x^*)<f(x^\prime)$ for all $x^\prime\in S^\prime$, where $S^\prime$ is the set of invalid solutions, so that $x^*$ occupies the global minimum of our objective function. It follows that:

\begin{equation}\label{gamma_star}
    \gamma>\gamma^* = \max_{x^\prime\in S^\prime}\Bigg(\frac{c(x^*)-c(x^\prime)}{p(x^\prime)}\Bigg)
\end{equation}

In general, for DQMs encoded as QUBOs, this quantity is not easily computable, as it requires that we find $x^*$, and we do not expect to find $x^*$ without evaluating a number of candidates exponential in the size of the solution vectors \cite{Ayodele2022}. Therefore, several heuristics for selecting a $\gamma$ that satisfies the inequality have been proposed, including setting it to the upper bound of the objective function, to the maximum QUBO coefficient, to the maximum change possible to the objective function in flipping a single bit, or using a sequential algorithm to find the minimum penalty for which the solution obtained is valid, starting from small $\gamma$ \cite{Verma2022, Ayodele2022, Garca2022}. However, it remains to be understood how the strength of the penalty parameter specifically influences the solution landscape and its optimization, provided satisfaction of Equation \ref{gamma_star}.

Qualitatively, given that the energies of valid solutions do not depend on $\gamma$ whereas those of invalid solutions do, as $\gamma$ increases from $\gamma^*$, peak height and valley depth in the solution landscape correspondingly increase. Steep QUBO solution landscapes are known to be difficult to traverse for classical algorithms such as simulated annealing \cite{palmer2018optimization}, but quantum algorithms (quantum annealing, QAOA) have been suggested to better navigate such spaces given their ability to tunnel through high energy barriers \cite{Chakrabarti2022, king2019quantum}. However, since quantum hardware implementations currently remain noisy, in combination with limited precision on qubit control, steep solution landscapes increase the failure rate of finding optimal solutions, by forcing valid solutions to occupy a narrower energy band \cite{Pearson2019, boothby2021architectural, xue2021effects}. Therefore, shallower QUBO solution landscapes at a first estimation should lend themselves to better solution by both classical and quantum approaches.

As a result, we would then set $\gamma = \gamma^*+\epsilon$, where $0<\epsilon\ll |\gamma^*|$, such that our landscape is as shallow as possible while still satisfying the condition that $x^*$ is the minimum energy solution out of all $x\in S \cup S^\prime$, $S$ being the set of valid solutions and $S^\prime$ the set of invalid solutions. We will maintain this notation throughout. However:

\begin{itemize}
    \item If we are interested in other low-energy valid solutions \cite{zucca2021diversity, pandey2022}, these solutions may not occupy local minima under such a penalty.
    \item We may be selecting for a large number of invalid solutions occupying low-energy local minima under such a penalty, which could negatively influence the performance of our search.
\end{itemize}

\noindent We say that $x_a$ is a local minimum if: $f(x_b)>f(x_a)$ for all $x_b$ such that $|x_b-x_a| = 1$. 

We note that while most classical optimization algorithms converge on local minima, this is not necessarily true of quantum optimization algorithms, which via quantum tunnelling are relatively ignorant of local solution landscape structure \cite{johnson2011quantum, Chakrabarti2022, mohseni2022ising}. However, because we desire a reasonable means way of preferentially sampling valid solutions over invalid solutions, it is often the case that a classical greedy descent post-processing algorithm is applied to quantum samples to bring them to local minima, under the assumption that these local minima are primarily valid solutions \cite{raymond2023hybrid, grant2022benchmarking}. 

While for $\gamma$ greater than some $\gamma^\dagger$ all $x^\prime\in S^\prime$ are with higher energy than all $x\in S$, for $\gamma < \gamma^\dagger$  there exists at least one $x^\prime$ with lower energy than some $x$. This simultaneously increases the possibility that a valid solution does not occupy a local minimum, and alters the possibility of invalid solutions occupying local minima due to solution landscape structure changes. These changes also shuffle the rank order of invalid solution energies as a result of the differences in slope of different invalid solutions with respect to $\gamma$, as may be seen in Figure \ref{fig_1}.

\begin{figure*}
    \centering
    \includegraphics[width=175mm]{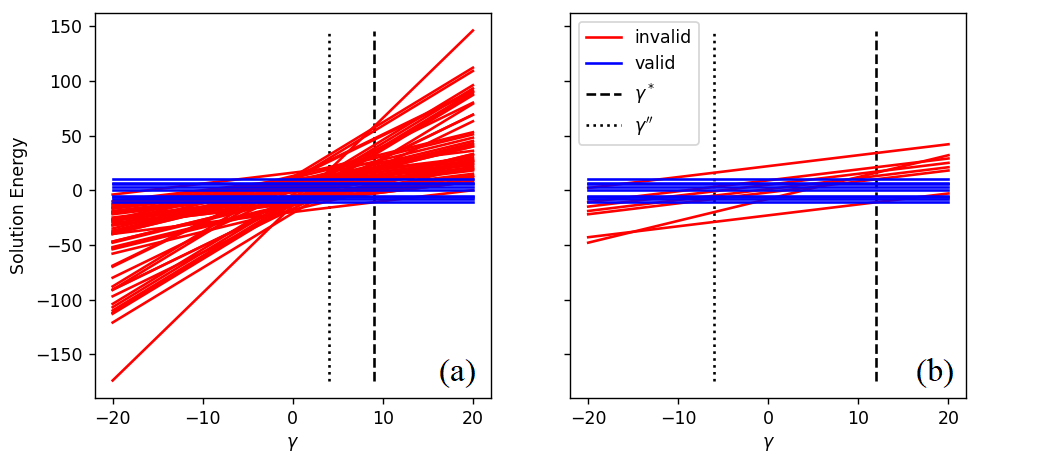}
    \caption{Valid and invalid solution energies as a function of $\gamma$ for a $k=2$, $l=3$ DQM, expressed as (a) a one-hot QUBO-encoded DQM, and (b) a domain-wall QUBO-encoded DQM. Notice that the solution energies of valid solutions are constant, while invalid solutions linearly increase in $\gamma$. Different invalid solutions exhibit different slopes and intercepts such that their rank order changes with $\gamma$, and it may be seen that steeper slopes are present in the one-hot QUBO-encoded DQM versus the domain-wall QUBO-encoded DQM. $\gamma^*$ is the penalty parameter which above which the optimal valid solution, $x^*$, occupies the global minimum. $\gamma^{\prime\prime}$ is the penalty parameter below which no valid solution occupies a local minimum.}
    \label{fig_1}
\end{figure*}

We now examine these solution landscape structure changes more rigorously. Sections \ref{one-hot} and \ref{domain-wall} begin by highlighting our findings for one-hot and domain-wall QUBO-encoded DQMs schemes, respectively, before presenting these findings in detail.

\subsection{One-hot encoding}\label{one-hot}

First, we show the existence of a problem-dependent $\gamma_{OH}^\prime$ such that for $\gamma_{OH}>\gamma_{OH}^\prime$ no invalid one-hot QUBO-encoded DQM solutions occupy local minima. Then, we show that $\gamma_{OH}^\prime$ does not necessarily equal $\gamma_{OH}^*$ by counterexample. Specifically, we demonstrate with this counterexample that $\gamma_{OH}^\prime>\gamma_{OH}^*$, indicating that whenever $\gamma_{OH}$ is large enough to isolate $x^*$ as the global minimum of $f_{OH}(x)$, it is possible that there exist invalid local minima. We then proceed to show the existence of a $\gamma_{OH}^{\prime\prime}$ such that for all $\gamma_{OH}<\gamma_{OH}^{\prime\prime}$, no valid solution occupies a local minimum. Using a similar argument, we also show that there exists a $\gamma_{OH}^{\prime\prime\prime}$ such that for all $\gamma_{OH}>\gamma_{OH}^{\prime\prime\prime}$, all valid solutions occupy local minima. Finally, we show by counterexample that $\gamma_{OH}^{\prime\prime}$ does not necessarily equal $\gamma_{OH}^*$. Specifically, we consider a case where $\gamma_{OH}^{\prime\prime} < \gamma_{OH}^*$, which indicates that when $\gamma_{OH}$ is large enough to isolate $x^*$ as the global minimum of $f_{OH}(x)$, there exist valid local minima other than that of $x^*$. We close this section with some comments on the implications of our findings for solution landscape navigability, more fully addressing these in Section \ref{discussion}.

Now, let us first establish that the change in the one-hot penalty between adjacent solutions $x_a$ and $x_b$ of Hamming distance one ($|x_b-x_a|=1$) is never zero. This property is important to our subsequent proofs, wherein this quantity appears as a denominator.

\begin{lemma}\label{lemma_1}
    Let $f_{OH}(x)=c(x)+\gamma_{OH} p(x)$ be a one-hot QUBO-encoded DQM function and $x_a$ and $x_b$ be neighboring solutions such that $|x_b-x_a|=1$. Then $p(x_b)-p(x_a)\neq 0$. 
\end{lemma}

\begin{proof}
    We first choose some register $i^\prime$ in which to flip a bit, i.e.~either $0\rightarrow 1$ or $1\rightarrow 0$. Then writing $\sum_\alpha b_{i, \alpha} = N_i$, we can express the penalty of $x_a$ in Equation \ref{penalty_oh} as:

    \begin{equation}
        p(x_a) = \sum_{i\neq i^\prime}(N_i-1)^2+(N_{i^\prime}-1)^2
    \end{equation}

    \noindent Similarly, we can express the penalty of $x_b$ as:

    \begin{equation}
        p(x_b) = \sum_{i\neq i^\prime}(N_i-1)^2+((N_{i^\prime}\pm1)-1)^2
    \end{equation}

    \noindent such that their difference, $p(x_a)-p(x_b)$ is:

    \begin{align}\label{penalty_diff_oh}
        p(x_a)-p(x_b) &= (N_{i^\prime}-1)^2-((N_{i^\prime}\pm1)-1)^2 \nonumber \\ &=
        \begin{cases}
            1-2N_{i^\prime}, & 0\rightarrow 1\text{ in register }i^\prime\\
            2N_{i^\prime}-3, & 1\rightarrow 0\text{ in register }i^\prime
        \end{cases}
    \end{align}

    \noindent Given that $N_{i^\prime}$ is an integer, we can conclude that $p(x_a)-p(x_b)\neq 0$ for all $ |x_b-x_a|=1$. 
\end{proof}

We can now show that there exists a $\gamma_{OH}^\prime$ such that for $\gamma_{OH}>\gamma_{OH}^\prime$, no invalid solution occupies a local minimum. In the proof that follows, we first show that for any particular invalid solution, there exists at least one $0\rightarrow 1$ or $1\rightarrow 0$ bit-flip that reduces the overall penalty of the solution. We then demonstrate that for all such bit-flips, we can guarantee that at least one per invalid solution involves a reduction in total energy (i.e., $f_{OH}(x_b)<f_{OH}(x_a)$, where we move from $x_a\rightarrow x_b$) by selecting an appropriately large value of $\gamma_{OH}$. Finally, we guarantee that this condition is simultaneously met for all invalid solutions by selecting the maximum of the set of $\gamma_{OH}$ values found in the previous step.

\begin{theorem}\label{theorem_oh_no_invalid_local_minima}
    Let $f_{OH}(x)=c(x)+\gamma_{OH} p(x)$ be a one-hot QUBO-encoded DQM function. Denote the set of valid solutions $S$ and the set of invalid solutions $S^\prime$. Then there exists a $\gamma_{OH}^\prime$ such that for $\gamma_{OH}>\gamma_{OH}^\prime$ there is guaranteed to exist an $x_b\in S\cup S^\prime$ such that $f_{OH}(x_b)<f_{OH}(x_a)$ for all $x_a\in S^\prime$ where $|x_b-x_a|=1$.
\end{theorem}

\begin{proof}
    For all $x_a\in S^\prime$, we claim that there must exist at least one 1-local neighbor $x_b\in S\cup S^\prime$ such that $f_{OH}(x_b)<f_{OH}(x_a)$:

    \begin{equation}\label{guaranteed_non-local_minima}
        c(x_b)+\gamma_{OH} p(x_b) < c(x_a)+\gamma_{OH} p(x_a)
    \end{equation}

    \noindent Isolating for $\gamma_{OH}$, this requires at least one of the following conditions to be true for some $x_b$ neighboring $x_a$:

    \begin{align}
        \gamma_{OH} &> \frac{c(x_b)-c(x_a)}{p(x_a)-p(x_b)},  &p(x_a)-p(x_b)>0 \label{first_condition_oh} \\
        \gamma_{OH} &< \frac{c(x_b)-c(x_a)}{p(x_a)-p(x_b)},  &p(x_a)-p(x_b)<0 \label{second_condition_oh}
    \end{align}

    \noindent Note that by Lemma \ref{lemma_1}, $p(x_a)-p(x_b)\neq0$, which avails us of having to consider the case where $p(x_a)-p(x_b)=0$.  

    We now show that for all $x_a\in S^\prime$ and $\gamma_{OH} > \gamma_{OH}^\prime$, at least one neighboring $x_b$ always satisfies the first condition (Equation \ref{first_condition_oh}). From this it follows that there exists an upper, finite bound to $\gamma_{OH}$, above which we are guaranteed to satisfy Equation \ref{guaranteed_non-local_minima} for all $x_a\in S^\prime$.

    When flipping a bit $0\rightarrow 1$ in the $i^\prime$ register of $x_a$, the difference in penalty between $x_a$ and $x_b$ is $1-2N_{i^\prime}$ (Equation \ref{penalty_diff_oh}). This expression we call $-\Delta p^+$:

    \begin{equation}\label{delta_p_plus}
        -\Delta p^+ =
        \begin{cases}
             1-2N_{i^\prime} <0, & N_{i^\prime} \neq 0\\
             1-2N_{i^\prime} >0, & N_{i^\prime} = 0
        \end{cases}
    \end{equation}

    \noindent Similarly, when flipping a bit $1\rightarrow 0$ in the $i^\prime$ register of $x_a$, we have the difference in penalty $2N_{i^\prime}-3$, which we call $-\Delta p^-$. This difference is such that:

    \begin{equation}\label{delta_p_minus}
        -\Delta p^- = 
        \begin{cases}
            2N_{i^\prime}-3 <0, & N_{i^\prime} = 1\\
            2N_{i^\prime}-3 >0, & N_{i^\prime} \neq 1
        \end{cases}
    \end{equation}

    We select an invalid register (a register $i^\prime$ whose vector $x_a(i^\prime)$ is such that $|x_a(i^\prime)|\neq 1$) in particular to flip a bit within, either $0\rightarrow 1$ or $1\rightarrow 0$. From Equations \ref{delta_p_plus} and \ref{delta_p_minus}, we can see that one of $-\Delta p^+$ or $-\Delta p^-$ is necessarily greater than zero if we select an appropriate bit-flip. That is, if $|x_a(i^\prime)|>1$ and we flip a bit from $1\rightarrow 0$, $-\Delta p^->0$, and if $|x_a(i^\prime)|<1$ and we flip a bit from $0\rightarrow 1$, $-\Delta p^+>0$. 

    At this stage, we have shown that for any $x_a\in S^\prime$ in particular, there exists a neighbor $x_b$ and a $\gamma_{OH}^\prime$ such that $f(x_b)<f(x_a)$ for $\gamma_{OH}>\gamma_{OH}^\prime$. We now identify the specific bound on $\gamma_{OH}$ that ensures that for all $x_a\in S^\prime$ there exists a neighbor $x_b$, where we require that $-\Delta p^{\pm} > 0$, such that $f_{OH}(x_b)<f_{OH}(x_a)$.

    For a specific invalid solution $x_a\in S^\prime$, the smallest $\gamma_{OH}$ possible that guarantees to admit at least one $x_b \in S\cup S^\prime$ such that $f_{OH}(x_b)<f_{OH}(x_a)$ is as follows:

    \begin{equation}
        \gamma_{OH} > \min_{x_b\in S\cup S^\prime}\Bigg(\frac{c(x_b)-c(x_a)}{p(x_a)-p(x_b)}\Bigg)
    \end{equation}

    \noindent where $|x_b-x_a|=1$ and $p(x_a)-p(x_b)>0$. Now, so that this is satisfied by one value of $\gamma_{OH}$ for all $x_a\in S^\prime$, $\gamma_{OH}$ must satisfy:

    \begin{equation}\label{gamma_oh_prime}
        \gamma_{OH} > \gamma_{OH}^\prime = \max_{x_a\in S^\prime}\Bigg\{\min_{x_b\in S\cup S^\prime}\Bigg(\frac{c(x_b)-c(x_a)}{p(x_a)-p(x_b)}\Bigg)\Bigg\}
    \end{equation}

    \noindent where the maximum is over the set of minimal $\gamma_{OH}$ values required to satisfy the existence of $x_a\rightarrow x_b$ transitions satisfying $f_{OH}(x_b)<f_{OH}(x_a)$ for all particular $x_a\in S^\prime$, again where $|x_b-x_a|=1$ and $p(x_a)-p(xb)>0$. Note that each value of this set is finite and well-defined; $|c(x_b)-c(x_a)|<\infty$ and $|p(x_a)-p(x_b)|\neq 0$. It follows that $\gamma_{OH}^\prime$ is finite and well-defined.
    
\end{proof}

We note again that $\gamma_{OH}^\prime$ represents an \textit{upper} bound to that threshold value of $\gamma_{OH}$ above which no invalid solution occupies a local minimum. That is, for $\gamma_{OH}<\gamma_{OH}^\prime$ it may be possible for all invalid solutions to not occupy local minima, supposing that Equation \ref{second_condition_oh} holds true where the Equation \ref{first_condition_oh} does not. However, we cannot guarantee that this will be always true, as the example proposed in Equation \ref{theorem_2} demonstrates.

We now consider this example in showing that $\gamma_{OH}^*$ does not equal $\gamma_{OH}^\prime$ in general. Ideally, we hope for this to be true, as it would indicate that as soon as the minimum-energy valid solution exists as the global minimum to the objective function in tuning $\gamma_{OH}$, no invalid solution occupies a local minimum. In such a case, it is then not be possible to sample invalid solutions with any quantum or classical algorithm that ends with a greedy descent step (such as the final stage of simulated annealing).

\begin{corollary}\label{theorem_oh_gammastar_neq_gammaprime}
    Let $\gamma_{OH}^*$ be the one-hot QUBO-encoded DQM penalty parameter such that for all $\gamma_{OH}>\gamma_{OH}^*$, $\min_{x\in S\cup S^\prime}f_{OH}(x) = f_{OH}(x^*)$. Let $\gamma_{OH}^\prime$ be the one-hot QUBO-encoded DQM penalty parameter such that for all $\gamma_{OH}>\gamma_{OH}^\prime$, no invalid solution $x_a\in S^\prime$ exists as a local minimum. It is not true, in general, that $\gamma_{OH}^*=\gamma_{OH}^\prime$.
\end{corollary}

\begin{proof}
    
    Consider the following one-hot QUBO-encoded DQM cost function, where $k = 2$ and $l=2$:

    \begin{equation}\label{theorem_2}
        c(x) = 
        \begin{bmatrix}
            3 & 0 & 2 & 4 \\
            0 & 3 & 1 & 2 \\
            0 & 0 & 4 & 0 \\
            0 & 0 & 0 & 7
        \end{bmatrix}
    \end{equation}

    \noindent Straightforward calculation of Equations \ref{gamma_star} and \ref{gamma_oh_prime} gives that $\gamma_{OH}^* = 5$ and $\gamma_{OH}^\prime = 6$.

\end{proof}

Additionally, for this example (Equation \ref{theorem_2}), there exists an invalid solution that occupies a local minimum when $\gamma_{OH}^* < \gamma_{OH} < \gamma_{OH}^\prime$, namely $(1 \ 0 \ 0 \ 0)$. This reinforces our claim that only when $\gamma_{OH}>\gamma_{OH}^\prime$ can we guarantee that no invalid solutions occupy local minima. Note that we generated this counterexample to the conjecture that $\gamma_{OH}^*=\gamma_{OH}^\prime$ by randomly sampling $k=2$, $l=2$ one-hot QUBO-encoded DQM instances, restricting their coefficients to the integers between 1 and 10. That a counterexample was found easily under these arbitrary restrictions suggests that such instances where $\gamma_{OH}^*\neq\gamma_{OH}^\prime$ are common in general.

We now turn our attention to the valid solutions, focusing on establishing bounds on $\gamma_{OH}$ above which all valid solutions are local minima, and below which no valid solution occupies a local minimum. We approach our proofs in a similar manner to that for Theorem \ref{theorem_oh_no_invalid_local_minima}, in this case considering our starting solutions $x_a$ to be valid.

\begin{theorem}\label{theorem 3}
    Let $f_{OH}(x) = c(x)+\gamma_{OH}p(x)$ be a one-hot QUBO-encoded DQM function. Denote the set of valid solutions $S$ and the set of invalid solutions $S^\prime$. Then there exists a $\gamma_{OH}^{\prime\prime}$ such that for $\gamma_{OH}<\gamma_{OH}^{\prime\prime}$ no $x_a\in S$ occupy local minima. Further, there exists a $\gamma_{OH}^{\prime\prime\prime}$ such that for $\gamma_{OH}>\gamma_{OH}^{\prime\prime\prime}$ all $x_a\in S$ occupy local minima.
\end{theorem}

\begin{proof}
     First, we recall that all bit-flips $0\rightarrow 1$ and $1\rightarrow 0$ from a given $x_a\in S$ move to invalid solutions $x_b\in S^\prime$ where $|x_b-x_a|=1$. Such moves come with an increase in penalty (i.e., $p(x_a)-p(x_b) = -1$), and as such, for $f_{OH}(x_b)<f_{OH}(x_a)$, from Equation \ref{second_condition_oh} we require:
    
    \begin{equation}\label{gamma_without_denominator}
        \gamma_{OH} < c(x_a)-c(x_b)
    \end{equation}
    
    \noindent for some neighboring $x_b$. That Equation \ref{gamma_without_denominator} is satisfied for at least one $x_b$ for a given $x_a\in S$, we take the maximum of this expression over all neighboring $x_b$:
    
    \begin{equation}
        \gamma_{OH} < \max_{x_b\in S^\prime}\big(c(x_a)-c(x_b)\big)
    \end{equation}
    
    \noindent We may then say that all valid solutions are not local minima under one penalty parameter $\gamma_{OH}$ if:
    
    \begin{equation}\label{gamma_oh_primeprime}
        \gamma_{OH}<\gamma_{OH}^{\prime\prime} = \min_{x_a\in S}\bigg\{\max_{x_b\in S^\prime}\big(c(x_a)-c(x_b)\big)\bigg\}
    \end{equation}
    
    \noindent where $|x_b-x_a|=1$. Similarly, all valid solutions are local minima when:
    
    \begin{equation}
        \gamma_{OH} > \gamma_{OH}^{\prime\prime\prime} = \max_{x_a\in S}\bigg\{\max_{x_b\in S^\prime}\big(c(x_a)-c(x_b)\big)\bigg\}
    \end{equation}
    
    \noindent where, necessarily, $\gamma_{OH}^{\prime\prime\prime}\geq\gamma_{OH}^{\prime\prime}$ and $|x_b-x_a|=1$.
    
\end{proof}

Whereas Theorem \ref{theorem_oh_no_invalid_local_minima} for invalid solutions establishes an upper bound to $\gamma_{OH}$ above which it is guaranteed that no invalid solutions occupy local minima, Theorem \ref{theorem 3} for valid solutions provides exact thresholds on $\gamma_{OH}$ above which all valid solutions are local minima and below which all valid solutions are not local minima. This exactness is a result of $-\Delta p^{\pm}$ always equalling $-1$ when starting from a valid solution and flipping a bit. 

We now show by counterexample that $\gamma_{OH}^{\prime\prime}$ does not equal $\gamma_{OH}^*$ in general, and that it is possible that $\gamma_{OH}^{\prime\prime}<\gamma_{OH}^*$. This indicates that as soon as the minimum-energy valid solution exists as the global minimum to the objective function through tuning of $\gamma_{OH}$, it is possible that other valid solutions already occupy local minima. 

\begin{corollary}
    Let $\gamma_{OH}^*$ be the one-hot QUBO-encoded DQM penalty parameter such that for all $\gamma_{OH}>\gamma_{OH}^*$, $\min_{x\in S\cup S^\prime}f_{OH}(x) = f_{OH}(x^*)$. Let $\gamma_{OH}^{\prime\prime}$ be the one-hot QUBO-encoded DQM penalty parameter such that for all $\gamma_{OH}<\gamma_{OH}^{\prime\prime}$ no valid solution $x_a\in S$ exists as a local minimum. It is not true in general that $\gamma_{OH}^*=\gamma_{OH}^{\prime\prime}$. More specifically, there exist cases where $\gamma_{OH}^{\prime\prime}<\gamma_{OH}^*$.
\end{corollary}

\begin{proof}
    
    Consider the following one-hot QUBO-encoded DQM cost matrix, where $k = 2$ and $l=2$:

    \begin{equation}\label{counterexample_2}
        c(x) = 
        \begin{bmatrix}
            7 & 0 & 5 & 4 \\
            0 & 7 & 5 & 9 \\
            0 & 0 & 2 & 0 \\
            0 & 0 & 0 & 6
        \end{bmatrix}
    \end{equation}

    \noindent Calculation of Equations \ref{gamma_star} and \ref{gamma_oh_primeprime} gives that $\gamma_{OH}^* = 12$ and $\gamma_{OH}^{\prime\prime} = 11$. The value of $\gamma_{OH}^{\prime\prime}$ stems from a valid solution $x\neq x^*$, indicating that this solution is a local minimum as soon as $x^*$ becomes the lowest-energy minimum within the solution energy landscape.

\end{proof}

Note that the above counterexample (Equation \ref{counterexample_2}) to the conjecture that $\gamma_{OH}^*=\gamma_{OH}^{\prime\prime}$ was generated by randomly sampling $k=2$, $l=2$ one-hot QUBO-encoded DQM instances, restricting their coefficients to the integers between 1 and 10. That such a counterexample was found easily under these arbitrary restrictions suggests that instances where $\gamma_{OH}^*\neq\gamma_{OH}^{\prime\prime}$ are common in general.

Commenting briefly on the applicability of the results achieved in this section, we remark that ideally-navigable QUBO-encoded DQM solution landscapes maximize the number of invalid solutions which do not occupy local minima, while simultaneously maximizing the number of non-interesting valid solutions which do not occupy local minima (maintaining $\gamma_{OH}>\gamma_{OH}^*$). We state these conditions as ideal given the well-understood point that a reduction in the number of minima in the solution landscape reduces the chances of heuristic search algorithms getting stuck in local minima (which in the case  of quantum algorithms pertains to classical post-processing of solutions). In certain cases (easily verified for $k=2$, $l=2$), it is indeed true that we can select $\gamma_{OH}$ such that the only local minimum of the full solution landscape is occupied by $x^*$, which means that from any solution we can greedily descend to this optimum, which is our desired best solution. A procedure for identifying in advance which one-hot QUBO-encoded DQM cases lend themselves to this solution landscape structure is an open question. 

Generally, in reducing $\gamma_{OH}$ from $\gamma_{OH}^{\prime\prime\prime}$, we see a decreasing number of valid solutions occupying local minima (with preferential preservation of low energy valid minima), while increasing $\gamma_{OH}$ from $\gamma_{OH}^*$ decreases the number of invalid solutions occupying local minima. For this reason, given that $\gamma_{OH}$ in the vicinity of $\gamma_{OH}^*$ minimizes the number of valid local minima and that low $\gamma_{OH}$ minimizes the jaggedness of the solution landscape, we suggest that such $\gamma_{OH}$ very likely encode solution landscapes whose search times are minimized for both classical and quantum approaches. However, to guarantee this requires further investigation.

\subsection{Domain-wall encoding}\label{domain-wall}

We start out with several qualitative observations of domain-wall QUBO-encoded DQM solution landscapes which are in contrast to their one-hot counterparts. First, whereas the 1-local neighbors of valid one-hot solutions are all invalid solutions, there exist between $k$ and $2k$ 1-local neighbors of a valid domain-wall solution that are also valid. In this sense there is a much different connectivity between solutions within the domain-wall solution landscape than the one-hot solution landscape. Secondly, given that the maximum penalty a domain-wall solution may incur is $k\lfloor(l-1)/2\rfloor$ versus $k(l-1)^2$ for the worst-case one-hot solution, the domain-wall energy landscape exists in a vertically-compressed form with respect to changing $\gamma$ versus the corresponding one-hot landscape (see Figure \ref{fig_1}). These features suggest that the structure of domain-wall QUBO-encoded DQM landscapes markedly differ from one-hot QUBO-encoded DQM solution landscapes, as we now demonstrate.

First, we show that no $\gamma_{DW}$ is sufficient to guarantee that no invalid solutions occupy local minima, where $\gamma_{DW}$ is the penalty parameter specific to domain-wall QUBO-encoded DQMs. We then proceed to show the existence of a $\gamma_{DW}^{\prime\prime}$ such that for all $\gamma_{DW}<\gamma_{DW}^{\prime\prime}$, no valid solution occupies a local minimum. Contrary to the one-hot case, we also show that there are no $\gamma_{DW}$ such that all valid solutions occupy local minima. Finally, we show by counterexample that $\gamma_{DW}^{\prime\prime}$ does not necessarily equal $\gamma_{DW}^*$. Specifically, we consider a case where $\gamma_{DW}^{\prime\prime} < \gamma_{DW}^*$, which indicates that when $\gamma_{DW}$ is large enough to isolate $x^*$ as the global minimum of $f_{DW}(x)$, there exist valid local minima other than that corresponding to $x^*$. We close this section with some comments on the implications of these findings for landscape navigability. We address these in more detail in Section \ref{discussion}.

Now, first we establish that the change in domain-wall penalty between adjacent solutions $x_a$ and $x_b$ of Hamming distance one ($|x_b-x_a|=1$) can be zero. This is important to our subsequent proofs.

\begin{lemma}\label{lemma_2}
    Let $f_{DW}(x)=c(x)+\gamma_{DW} p(x)$ be a domain-wall QUBO-encoded DQM function and $x_a$ and $x_b$ be neighboring solutions such that $|x_b-x_a|=1$. Then it is possible that $p(x_a)-p(x_b) = 0$. 
\end{lemma}

\begin{proof}
    We first choose some register $i^\prime$ in which to flip a bit either $0\rightarrow 1$ or $1\rightarrow 0$. We then express the penalty of $x_a$ as:

    \begin{align}
        p(x_a) &= \sum_{i\neq i^\prime}\sum_{\alpha}(b_{i, \alpha}-b_{i, \alpha}b_{i, \alpha-1}) \nonumber \\& \hspace{17.5pt}+\sum_{\alpha}(b_{i^\prime, \alpha}-b_{i^\prime, \alpha}b_{i^\prime, \alpha-1}) \nonumber \\
        &=\sum_{i\neq i^\prime}(N_i-1)+(N_{i^\prime}-1)
    \end{align}

    \noindent where $N_i$ and  $N_{i^\prime}$ are the number of domain walls present in registers $i$ and $i^\prime$ (valid registers being with one domain-wall present). Similarly, we express the penalty of $x_b$ as:

    \begin{equation}
        p(x_b) = \sum_{i\neq i^\prime}(N_i-1)+(N_{i^\prime}^\prime-1)
    \end{equation}

    \noindent where $N_{i^\prime}^\prime\in\{N_{i^\prime}-1, N_{i^\prime}, N_{i^\prime}+1\}$. To see this, consider the register $i^\prime$ of length $l-1=7$ with the register:

    \begin{equation}
        x_a(i^\prime) = \texttt{1 0 1 0 0 0 1 }
    \end{equation}

    \noindent Indexing from zero, we see that to flip a bit $0\rightarrow 1$ at position 1 removes a domain wall ($N_{i^\prime}^\prime=N_{i^\prime}-1$), to flip a bit $0\rightarrow 1$ at position 3 extends a domain wall ($N_{i^\prime}^\prime=N_{i^\prime}$), and to flip a bit $0\rightarrow 1$ at position 4 adds a domain wall ($N_{i^\prime}^\prime = N_{i^\prime}+1$). Therefore, we have:

    \begin{equation}\label{delta_p_dw}
        p(x_a)-p(x_b) = N_{i^\prime}-N_{i^\prime}^\prime \in \{-1, 0, +1\}
    \end{equation}
\end{proof}

We now show that, unlike for one-hot QUBO-encoded DQMs, there does not exist, in general, an upper bound to $\gamma_{DW}$ above which no invalid solutions occupy local minima for domain-wall QUBO-encoded DQMs. In the proof that follows, we first show that there exists a class of invalid solutions whose elements are such that for all transitions to neighbors of a given element within this class, no change in penalty is incurred. We then show that for such solutions to exist as local minima, there must be at least one neighbor whose cost is lower than the starting solution. Finally, we show that this condition cannot always be satisfied.

\begin{theorem}\label{theorem_dw_no_invalid_local_minima}
    Let $f_{DW}(x)=c(x)+\gamma_{DW} p(x)$ be a domain-wall QUBO-encoded DQM function. Denote the set of valid solutions $S$ and the set of invalid solutions $S^\prime$. Then there is not guaranteed to exist a $\gamma_{DW}^\prime$ such that for $\gamma_{DW}>\gamma_{DW}^\prime$ there is at least one $x_b\in S\cup S^\prime$ such that $f_{DW}(x_b)<f_{DW}(x_a)$ for all $x_a\in S^\prime$ where $|x_b-x_a|=1$.
\end{theorem}

\begin{proof}
    To show that this is true, we assume the contrary, namely, that there exists a $\gamma_{DW}^\prime$ such that for $\gamma_{DW}>\gamma_{DW}^\prime$ there is at least one $x_b\in S\cup S^\prime$ such that $f_{DW}(x_b)<f_{OH}(x_a)$ for all $x_a\in S^\prime$ where $|x_b-x_a|=1$.

    This requires that for all $x_a\in S^\prime$, there is at least one $x_b$ such that $|x_b-x_a|=1$ satisfying:

    \begin{equation}
        c(x_b)+\gamma_{DW}p(x_b)<c(x_a)+\gamma_{DW}p(x_a)
    \end{equation}

    Unlike in the one-hot case, where $p(x_a)-p(x_b)\neq 0$ allows us to solve for $\gamma_{OH}$ as in Equations \ref{first_condition_oh} and \ref{second_condition_oh}, for a DQM problem expressed as a domain-wall encoded QUBO it is possible to have $p(x_a)-p(x_b)=0$, as per Lemma \ref{lemma_2}. Given this, for a particular $x_a\in S^\prime$ to have at least one neighbor $x_b:|x_b-x_a|=1$ such that $f_{DW}(x_b)<f_{DW}(x_a)$, there must exist an $x_b$ or $(x_b, \gamma_{DW})$ pair satisfying one of the following cases:

    \begin{equation}\label{gamma_dw}
        \begin{aligned}
            c(x_b)<c(x_a),\text{    } & N_{i^\prime}^\prime=N_{i^\prime} \\
            \gamma_{DW}>c(x_b)-c(x_a),\text{    } & N_{i^\prime}^\prime=N_{i^\prime}+1 \\
            \gamma_{DW}<c(x_a)-c(x_b),\text{    } & N_{i^\prime}^\prime=N_{i^\prime}-1
        \end{aligned}
    \end{equation}

    We now note that there exist 5 distinct classes of invalid register, some for which certain of the above cases do not ever apply when a solution consists of only these kinds of register or their combination with valid registers. These classes are as follows:

    \begin{equation}\label{classes}
        \begin{array}{lll}
            A: & -\Delta p\in\{-1, 0, +1\} & \text{e.g., }\texttt{1 0 1 0 0 0 1} \\
            B & -\Delta p\in\{0, +1\} & \text{e.g., }\texttt{1 0 1 1 0 }\\
            C: & -\Delta p\in\{-1, 0\} & \text{e.g., }\texttt{1 1 0 0 0 1 1}\\
            D: & -\Delta p = 0 & \text{e.g., }\texttt{0 0 1 1 0 }\\
            E: & -\Delta p = +1 & \text{e.g., }\texttt{0 1}
        \end{array}\nonumber
    \end{equation}

    \noindent  That a class where $-\Delta p = -1$ only (amounting to a class of solutions where we can only \textit{add} a domain-wall) does not exist can be easily verified.
    
    This means that only classes $A$, $B$, and $E$ contain invalid registers that are able of losing a domain-wall upon a single bit-flip. Therefore, a $\gamma_{DW}$ may always be selected for solutions containing only these classes of register (or in combination with valid registers) such that there exists an $x_b$ allowing the satisfaction of:

    \begin{equation}
        \gamma_{DW}>c(x_b)-c(x_a)
    \end{equation}

    \noindent where $N_{i^\prime}^\prime=N_{i^\prime}+1$, which is the second case of Equation \ref{gamma_dw}. Specifically, for a particular invalid solution $x_a$ containing either of a class $A$, class $B$, or class $E$ register, the smallest $\gamma_{DW}$ that admits this inequality is:

    \begin{equation}
        \gamma_{DW}>\min_{x_b\in S\cup S^\prime}\big(c(x_b)-c(x_a)\big)
    \end{equation}

    \noindent such that if one $\gamma_{DW}$ is to suffice for all particular instances of such $x_a$:

    \begin{equation}
        \gamma_{DW}>\gamma_{DW}^{\prime}=\max_{x_a\in S^\prime}\Bigg\{\min_{x_b\in S\cup S^\prime}\big(c(x_b)-c(x_a)\big)\Bigg\}
    \end{equation}

    Shifting our attention to those invalid solutions that do not contain a class $A$, class $B$, or class $E$ register, we consider solutions that either admit only $-\Delta p \in\{-1, 0\}$ (class $C$) or $-\Delta p = 0$ (class $D$). Focusing on class $D$ invalid solutions, which form a part of the solution space for domain-wall DQMs whenever $l-1\geq 4$, we see that for the elements of this class to not occupy local minima, there must exist for a particular $x_a$ in this class at least one 1-local neighbor $x_b$ such that $c(x_b)<c(x_a)$. This we cannot guarantee in general; considering a domain-wall QUBO-encoded DQM of one register with length $l-1 = 4$ satisfies to show this rather trivially. 
    
    We therefore conclude that in general there does not exist a $\gamma_{DW}$ that is guaranteed to ensure that all $x_a\in S^\prime$ do not occupy local minima for a domain-wall encoded DQM. However, in certain cases it is possible that this condition might be satisfied.
    
\end{proof}

We just established that for domain-wall QUBO-encoded DQMs it cannot be guaranteed that there is an upper bound on $\gamma_{DW}$ ($\gamma_{DW}^\prime$) above which all invalid solutions are not local minima. Further, we identified a subset of invalid solutions such that to increase $\gamma_{DW}$ decreases the number of solutions within this subset that occupy local minima (classes $A$, $B$ and $E$ in particular, where $-\Delta p^\pm$ can equal $+1$). For another subset of invalid solutions (classes $A$ and $C$), decreasing $\gamma_{DW}$ decreases the number of solutions within this subset that occupy local minima. Finally, we identify a subset of invalid solutions for which $\gamma_{DW}$ has no effect on whether they occupy local minima or not (class $D$), meaning that if one of this class happens to occupy a local minimum, it will do so irrespective of $\gamma_{DW}$. 


We now turn our attention to valid solutions, seeking the bound to $\gamma_{DW}$ below which all valid solutions do not occupy local minima. We find that in contrast to one-hot QUBO-encoded DQM solution landscapes, we cannot find a bound to $\gamma_{DW}$ above which all valid solutions do occupy local minima. Our proof proceeds in a manner similar to our approach to Theorem \ref{theorem_dw_no_invalid_local_minima}.

\begin{theorem}\label{theorem 6}
    Let $f_{DW}=c(x)+\gamma_{DW}p(x)$ be a domain-wall QUBO-encoded DQM function. Denote the set of valid solutions $S$ and the set of invalid solutions $S^\prime$. Then there exists a $\gamma_{DW}^{\prime\prime}$ such that for $\gamma_{DW}<\gamma_{DW}^{\prime\prime}$ no $x_a\in S$ occupy local minima. Further, there does not exist a $\gamma_{DW}^{\prime\prime\prime}$ such that all $x_a\in S$ occupy local minima for $\gamma_{DW}>\gamma_{DW}^{\prime\prime\prime}$.
\end{theorem}

\begin{proof}
    First, we point out that all single bit-flip moves away from a given $x_a\in S$ move to either invalid solutions that contain one additional domain-wall as compared to $x_a$, or to valid solutions having the same number of domain-walls as $x_a$. (All valid solutions have $k$ total domain-walls, uniformly distributed across $k$ registers.)

    Calling $x_b$ the solution moved to upon a single bit-flip from $x_a$, if $x_b\in S^\prime$, it is clear from Equation \ref{delta_p_dw} that $p(x_a)-p(x_b)=-1$. Similarly, if $x_b\in S$, we have that $p(x_a)-p(x_b)=0$. Therefore, for a particular $x_a\in S$ to have at least one neighbor $x_b$ where $|x_b-x_a|=1$ such that $f_{DW}(x_b)<f_{DW}(x_a)$, there must exist an $x_b$ or $(x_b, \gamma_{DW})$ pair satisfying either:

    \begin{equation}\label{gamma_cases_dw}
        \begin{aligned}
            c(x_b)<c(x_a),\text{    } &x_b\in S \\
            \gamma_{DW} < c(x_a)-c(x_b),\text{    } &x_b\in S^\prime
        \end{aligned}
    \end{equation}

    \noindent That the first case is satisfied is true except for when $c(x_a)<c(x_b)$ $\forall x_b\in S:|x_b-x_a|=1$. Therefore, to guarantee that all $x_a\in S$ do not occupy local minima, we focus on making sure that the second case holds true for such $x_a$, where $c(x_a)<c(x_b)$ $\forall x_b\in S:|x_b-x_a|=1$. Now, for a particular such $x_a$, that there exists at least one $x_b\in S^\prime$ satisfying this expression (case 2, Equation \ref{gamma_cases_dw}), we take the maximum over all neighboring $x_b\in S^\prime$:

    \begin{equation}\label{double_prime_dw}
        \gamma_{DW}<\max_{x_b\in S^\prime}\big(c(x_a)-c(x_b)\big)
    \end{equation}

    \noindent We may then say that all valid solutions are not local minima under one penalty parameter $\gamma_{DW}$ if:

    \begin{equation}
        \gamma_{DW}<\gamma_{DW}^{\prime\prime}=\min_{x_a\in S}\Bigg\{\max_{x_b\in S^\prime} \big(c(x_a)-c(x_b)\big)\Bigg\}
    \end{equation}
    
    \noindent where $x_a:c(x_a)<c(x_b)$ $\forall x_b\in S:|x_b-x_a|=1$.

    Considering again the cases of Equation \ref{gamma_cases_dw}, we note that irrespective of our choice of $\gamma_{DW}$, there must exist at least one $x_a\in S$ with a set of 1-local neighbors in $S$ such that for a particular $x_b$ in this set $c(x_b)<c(x_a)$ is satisfied. If not, then all $x_a\in S$ would be such that $c(x_a)<c(x_b)$ $\forall x_b\in S:|x_b-x_a|=1$, which implies a contradiction unless $\forall x_a, x_b\in S$, $c(x_a)=c(x_b)$, which corresponds to the trivial DQM. Therefore, there does not exist a $\gamma_{DW}^{\prime\prime\prime}$ above or below which all valid solutions occupy local minima.   
\end{proof}

We now show by counterexample that $\gamma_{DW}^{\prime\prime}\neq\gamma_{DW}^*$ in general, and indeed that it is possible that $\gamma_{DW}^{\prime\prime}<\gamma_{DW}^*$, which indicates that as soon as the condition is achieved that the minimum-energy valid solution exists as the global minimum to the objective function, it is possible that other valid solutions might already occupy local minima. 

\begin{corollary}
    Let $\gamma_{DW}^*$ be the domain-wall QUBO-encoded DQM penalty parameter such that for all $\gamma_{DW}>\gamma_{DW}^*$, $\min_{x\in S\cup S^\prime}f_{DW}(x) = f_{DW}(x^*)$. Let $\gamma_{DW}^{\prime\prime}$ be the domain-wall QUBO-encoded DQM penalty parameter such that for all $\gamma_{DW}<\gamma_{DW}^{\prime\prime}$ no valid solution $x_a\in S$ exists as a local minimum. It is not true in general that $\gamma_{DW}^*=\gamma_{DW}^{\prime\prime}$. More specifically, there exist cases where $\gamma_{DW}^{\prime\prime}<\gamma_{DW}^*$.
\end{corollary}

\begin{proof}

    Consider the following domain-wall QUBO-encoded DQM cost function, where $k = 2$ and $l = 3$:

    \begin{equation}\label{counterexample_3}
        c(x) = 
        \begin{bmatrix}
            4 & -2 & 3 & 1 \\
            0 & 1 & 2 & -2 \\
            0 & 0 & -1 & 4 \\
            0 & 0 & 0 & -4
        \end{bmatrix}-5        
    \end{equation}

    \noindent Calculation of Equations \ref{gamma_star} and \ref{double_prime_dw} gives that $\gamma_{DW}^* = 3$ and $\gamma_{DW}^{\prime\prime} =-3$, concluding the proof.
    
\end{proof}

We note that the above counterexample (Equation \ref{counterexample_3}) to the conjecture that $\gamma_{DW}^*=\gamma_{DW}^{\prime\prime}$ was generated by randomly sampling $k=2$, $l=3$ domain-wall QUBO-encoded DQM instances, restricting their coefficients to the integers between $-9$ and $9$. That such a counterexample was found easily under these arbitrary restrictions suggest that such instances where $\gamma_{DW}^*>\gamma_{DW}^{\prime\prime}$ are common in general.

We now comment briefly on the applicability of the results of this section. As noted in the one-hot QUBO-encoded DQM section above, ideally-navigable QUBO-encoded DQM solution landscapes maximize the number of invalid solutions not occupying local minima, and simultaneously maximize the number of non-interesting valid solutions not occupying local minima, maintaining $\gamma_{DW}>\gamma_{DW}^*$. Unlike in the one-hot QUBO-encoded DQM case, here we are unable to claim with certainty that the number of invalid solutions either increases or decreases with increasing $\gamma_{DW}$, beginning from $\gamma_{DW}^*$. However, we can say that the number of valid solutions occupying local minima decreases if $\gamma_{DW}\gg\gamma_{DW}^{\prime\prime}$ and is subsequently reduced toward $\gamma_{DW}^{\prime\prime}$. In certain cases (easily verified by selecting from random $k=2$, $l=2$ instances), indeed, we can select $\gamma_{DW}$ such that the only local minimum of the full solution landscape is occupied by $x^*$. This means that from any solution we can greedily descend to this optimum, which is our desired best solution. We are currently not aware of a procedure for identifying \textit{a priori} DQM problems that may be encoded this way.

\section{Discussion}\label{discussion}

The task of solving DQMs as QUBO models given the availability of increasingly-powerful, special-purpose QUBO solvers (including quantum annealers, the QAO algorithm
executed by gate-model quantum computers, and digital annealers) motivated our investigation into the structure of their solution landscapes upon one-hot or domain-wall encoding. These encodings are such that they demonstrate large structural differences in comparison to one another. In this section, we discuss these differences, and point out their various benefits and shortcomings.

To begin, recall that a one-hot encoding of a $k$-variable DQM with $l$ possible values per variable requires $kl$ binary variables, and $kl(kl-1)/2$ pairwise interactions between these variables. Considering a domain-wall encoding of the same problem, only $k(l-1)$ binary variables and $k(l-1)(k(l-1)-1)/2$ pairwise interactions are required. This savings in the number of variables and interactions has the corollary that there are $2^{kl}(1-2^{-k})$ less invalid solutions in a domain-wall QUBO-encoded DQM solution landscape versus a one-hot QUBO-encoded DQM solution landscape. Based on this, a domain-wall QUBO-encoded DQM is more desirable compared to a one-hot QUBO-encoded DQM if we are concerned with spatial resources, such as the number of qubits available on a quantum annealer. (Also, it is generally true that smaller problem sizes of the same kind are solved more quickly with both classical and quantum optimization techniques.)


However, given the fact that domain-wall QUBO-encoded DQM matrix entries involve sums of one-hot QUBO-encoded DQM matrix entries, the largest entries in the domain-wall case can be larger than the largest entries in the one-hot case. The importance of this concerns the limited dynamic range of quantum annealer bias and coupling devices in particular, which may be thought of as the physical instantiations of QUBO matrix entries. The entries of a QUBO matrix must be made to fit within the dynamic range of the bias and coupling devices, and the larger the entries of the QUBO matrix, the more rescaling of these values must take place to ensure this fit, which may be compromising in the face of integrated control errors and noise \cite{Pearson2019}. This suggests that one-hot encodings may better lend themselves to noisy solvers versus domain-wall encodings. However, Berwald et al. note that while domain-wall-encoded DQMs do have larger matrix entries than their one-hot-encoded counterparts (and therefore higher chain strengths), domain-wall encodings sampled by quantum annealing exhibit lower effective temperatures given the problem they consider, and favorable results in comparison to one-hot encoding \cite{Berwald2022}.

Another potential drawback to domain-wall QUBO-encoded DQMs results from the connectivity of the solution space, which allows any valid solution to be transformed into any other valid solution without having to pass through an invalid solution. As a particular example of why this might be detrimental, we consider molecular docking, a structural biology problem. Here, we aim to find low-energy configurations of a set of molecules in a discretized space, typically in 2 or 3 dimensions. The points within this space form the number of choices within our registers, and one molecule is assigned per register. Note that the multi-dimensional real space of the problem is encoded to a higher-dimensional binary space when transformed into a QUBO. This change in topology can place solutions that are far from one another in the real space of the problem exactly adjacent to one another upon domain-wall QUBO-encoding of a DQM. When we are interested in multiple low-energy solutions, it is then possible that a low-energy valid solution within the real space of the problem will not occupy a local minimum in the solution landscape of the QUBO-encoded DQM, by virtue of its being adjacent to another low-energy valid solution.


By contrast, one-hot QUBO-encoded DQMs allows full separation between valid solutions of interest if $\gamma_{OH}$ is selected so that no invalid solution is with an energy below those valid solutions of interest. This is a consequence of Theorems \ref{theorem 3} and \ref{theorem 6}. It is difficult to know in advance if a problem may suffer from a situation similar to that just described; in any case, one has to be wary of this only when more than one valid solution is sought.

Another difference between one-hot and domain-wall QUBO-encoded DQMs concerns whether invalid solutions occupy local minima or not as a function of the penalty parameter. In the one-hot case, we proved in Theorem \ref{theorem_oh_no_invalid_local_minima} that we can select a $\gamma_{OH}$ that guarantees that no invalid solutions occupy local minima, whereas in the the domain-wall case we proved in Theorem \ref{theorem_dw_no_invalid_local_minima} that we cannot guarantee this through selection of any $\gamma_{DW}$. This means that with a one-hot encoding, we can ensure that we sample only valid solutions provided we have selected sufficiently large $\gamma_{OH}$, whereas we cannot avoid the possibility of sampling invalid solutions with a domain-wall encoding in advance. Our results do not suggest how many such invalid solution local minima are unavoidable with a domain-wall encoding, but given that there are $(l-3)^k$ possible solutions comprised of only class $D$ registers, their number may be substantial (as even a solution with one class $D$ register, the remaining registers being valid, may occupy a local minimum under all $\gamma_{DW}$). This remains to be understood rigorously.

Now, apart from our observations regarding solution landscape structures between one-hot and domain-wall QUBO-encoded DQMs, we also point out the following consequence of their encoding structures. Namely, domain-wall QUBO-encoded DQMs lack the ability to efficiently leverage a certain kind of symmetry that might be present in a DQM problem, whereas one-hot QUBO-encoded DQMs can be adapted to accommodate these symmetries in ways which prove more compact than a domain-wall encoding. Again with reference to molecular docking, we can imagine a situation in which we are with $k$ copies of some molecule $M$, to be docked in a space containing $l>k$ points. Usually, we assign one register per molecule in a one-hot encoding, but in this case we are permitted a representation using just one register of $l$ points, requiring that $k$ bits are one for valid solutions, corresponding to the unique placement of the $k$ molecules. In this case, the one-hot penalty of Equation \ref{penalty_oh} becomes a $k$-hot penalty:

\begin{equation}
    H_{KH}^P=\Big(\sum_\alpha b_{i,\alpha}-k\Big)^2
\end{equation}

\noindent over a single register. More generally put, $k$-hot encoding applies to a DQM whenever multiple variables are with the same support and interactions with other variables, and are constrained to select different choices from their support. By contrast, a domain-wall encoding scheme cannot be adapted so readily in this way. Namely, whereas a register of length $l$ can accommodate $l$ bits set to $1$, it can only accommodate $\lceil l/2 \rceil$ domain-walls, under the restriction we have been working with that these domain walls are represented by \texttt{10}, and not \texttt{01}. Given this, certain problems are more compactly represented by an adaption of a one-hot encoding scheme ($k$-hot), or a combination of the $k$-hot and domain-wall encoding schemes, versus the domain-wall encoding scheme. 

Finally, we close our discussion of solution landscape structures in remarking of our proofs of various thresholds on $\gamma_{OH}$ and $\gamma_{DW}$ that they are proofs of existence, and non-constructive. That is, though these thresholds are well-defined in terms of max-min, min-max, or max-max procedures, they generally require evaluating the costs of all valid solutions. Naively calculating these thresholds is therefore extremely difficult, which we might seek to know in order to produce a best-navigable solution landscape in advance. However, this does not preclude their efficient calculation in special cases, nor their efficient approximation in the general case.

\section{Conclusions}\label{conclusions}

Solving a discrete quadratic model as a QUBO model requires translating the DQM to this form, commonly via one-hot or domain-wall encoding. Both encodings introduce invalid solutions to the solution space, and a parameter to penalize these solutions. The solution spaces of each encoding differ in the connectivity between valid solutions, the distribution of local minima, and their response to changing penalty strength. We have conducted a preliminary investigation of these differences, noting the shifting structure of local minima relative to penalty parameter strength, and finding that best selection between a one-hot and domain-wall encoding is problem-dependent.

This work emphasizes the importance of penalty parameter and encoding choice to QUBO-encoded DQM solution landscapes and optimization. Specifically, we find that for one-hot QUBO-encoded DQMs, no invalid solution occupies a local minimum and all valid solutions occupy local minima for sufficiently large penalty strengths. For domain-wall QUBO-encoded DQMs, we cannot in general guarantee no invalid solutions will occupy local minima, nor that all valid solutions occupy local minima.  Given this, and that the differing connectivities and dimensions of one-hot and domain-wall encodings lend themselves differently to answering questions where a diverse set of low-energy valid solutions are sought, we suggest that DQMs need to be encoded as QUBO models on a problem-by-problem basis.


\section*{Acknowledgements}\label{Acknowledgements}

This work was supported in part by a NSERC CREATE in quantum computing grant (543245-2020), as well as an NSERC CGS-M scholarhip (Zaborniak). The research contributions of this article were initiated by the lead author. Both authors discussed the research. The second author contributed to revisions of the manuscript.

\bibliographystyle{IEEEtran}
\bibliography{IEEEabrv, references}

\end{document}